\documentclass[journal]{IEEEtran}
\ifCLASSINFOpdf

\else

\fi
\usepackage{amsmath}
\usepackage{makeidx}  
\usepackage{algorithm}
\usepackage{algorithmic}
\usepackage{graphicx}
\usepackage{subfigure}
\usepackage{epstopdf}
\usepackage{bm}
\usepackage{cite}
\usepackage{stfloats}

\usepackage{amssymb}
\usepackage{graphicx}

\usepackage{url}

\usepackage{tabularx,booktabs}
\newcolumntype{C}{>{\centering\arraybackslash}X} 
\usepackage{lipsum}

\usepackage{makecell} 

\usepackage{graphicx}
\usepackage{color}
\newtheorem{thm}{Theorem}

\newtheorem{proof}{proof}

\begin{document}

\title{Throughput  Analysis of UAV-assisted Cellular Networks by Mat\'{e}rn Hardcore Point Process}


\author{\IEEEauthorblockN{Mengbing~Liu, Guangji~Chen, and Ling~Qiu}\\
\thanks{Mengbing Liu, Guangji Chen, and Ling Qiu are with Key Laboratory of Wireless-Optical Communications, Chinese Academy of Sciences, School of Information Science and Technology, University of Science and Technology of China (email: liumb@mail.ustc.edu.cn)}

}

\maketitle

\begin{abstract}
Unmanned aerial vehicles (UAVs) are expected to coexist with conventional terrestrial cellular networks and become an important component to support high rate transmissions. This paper presents an analytical framework for evaluating the throughput performance of a downlink two-tier heterogeneous network. Considering the minimum distance constraint among UAVs, Mat\'{e}rn hardcore point process (MHP) is utilized to model the locations of UAVs. The locations of terrestrial base stations (BSs) are modeled by Poisson point process (PPP). Tools of stochastic geometry are invoked to derive tractable expressions for average data rates of users. With the analytical results, we discuss the optimal combinations of UAVs' height and power control factor. The result shows that an appropriate power control factor can effectively maximize UAV users' average data rate as well as guaranteeing the BS users' performance under our proposed model.
\end{abstract}


\IEEEpeerreviewmaketitle

\vspace{-6pt}
\section{Introduction}
Unmanned aerial vehicles (UAVs) have attracted growing interest to provide short-term wireless communication links to ground users in some temporary or unexpected scenarios, where excessive need for network resources is required such as festival, rescue, etc. Thus, deploying UAVs as well as small-cell base stations (BSs) appears to be a promising approach to meet the ever increasing data demand and balance the heavy load of macro BSs, of which the system called UAV-assisted cellular networks \cite{Zeng2019Accessing}, \cite{Tutorial2019MM}.

For UAV-assisted cellular system, modeling their locations as a random point process by stochastic geometry (SG) is an effective method to evaluate the network level performance \cite{j2011atractable}. Most of the existing work regarding UAV-assisted communication system modeled the locations of UAVs and BSs as Poisson point process (PPP) due to its tractability [4]-[6].  The coverage probability with user-centric strategy and UAV-centric strategy have been investigated in NOMA assisted multi-UAV framework in \cite{NOMA2019Liu}. The optimal UAV density has been obtained to maximize the downlink throughput in a two-tier system in \cite{c2017spectrum}. The coverage probability  of a two-tier downlink network has been studied in hot spots in \cite{e2018downlink}. However, PPP cannot characterize UAVs' deployment accurately in some specific scenarios and adopting the suitable topology model is of great importance for performance analysis \cite{drone2016am}-[9].  For example, deploying UAVs to malfunction areas where ground BSs cannot work properly is an important application. For such a scenario, the UAVs are usually deployed in a given area. To capture this feature, the Binomial point process (BPP) was employed in \cite{drone2016am} to model the locations of UAVs in a disaster recovery scenario.  Considering that wireless systems possess a clustered property in the hot spots, Poisson cluster process (PCP) was adopted to  model the locations of UAVs and the optimal height of UAVs was derived to maximize the system coverage probability in \cite{A2018wyi}. What's more, considering some secure issues, UAVs are not allowed to be closer than a certain distance in reality, which suggests Mat\'{e}rn hardcore point process (MHP) is more suitable to model the locations of UAVs and secrecy rate has been analyzed in a single tier UAV network modeled by MHP in \cite{secrecy2018yz}.

While previous works in evaluating network performance of UAV-assisted networks have provided solid design insights, blanks are still left unfilled. This paper is motivated by the following observations: 1) MHP has been already used to model the locations of UAVs in \cite{secrecy2018yz}.  However, the system
network is a single tier UAV network . The analytical result of a two-tier UAV-assisted cellular network with co-channel interference between BS tier and UAV tier has not been investigated. What's more, the Laplace function of  aggregate interference is complicated which is caused by the repulsion property among points and additional inter-tier interference in a two-tier MHP-based network. 
 Thus, a more tractable expression for performance analysis is necessary to be explored.
 2) Due to the line-of-sight (LoS) dominated channel characteristic, UAVs'air-to-ground (A2G) channels will inevitably generate severe interference to the BS users (BUEs) while considering a two-tier co-channel network. \cite{c2017spectrum} applied primary exclusive regions to reduce the inter-tier interference to cellular users which is caused by UAVs. \cite{e2018downlink}-\cite{drone2016am} didn't consider interference management to mitigate the inter-tier interference. Thus, in order to improve system performance, it is necessary to utilize effective interference management to mitigate the severe inter-tier interference to BUEs. These two observations advocate that developing a systematic mathematically tractable framework to facilitate performance evaluation of MHP-based UAV-assisted networks by incorporating interference management is desired. To the best of our knowledge, none of the previous researches jointly consider the aforementioned two issues in the analysis.

In this paper,  MHP is utilized to model the locations of UAVs
in which points are not allowed to be closer to each other than a certain minimum distance\cite{am2013coverage} and the locations of terrestrial BSs are modeled as PPP. We consider employing power control to mitigate the interference to BUEs caused by UAVs' transmission via A2G links. Besides this, the Laplace function of the aggregate interference in MHP-based networks is much more complicated than that in PPP-based networks, which makes the analysis challenging. To address this problem, \cite{m2014themean} introduced a simple yet powerful analytical framework for a tight SIR distribution approximation in general non-Poisson networks which is called mean interference-to-signal ratio (MISR)-based gain method. However, the work above just proposed a general MISR framework, the exact expression of MISR-based gain in MHP-based network is still unexplored. Aiming at tackling the issue, the main contributions of this paper are summarized as follows:

1) With the aid of MISR-based gain method , we exploit the approximate relationship for distributions of SIR to analyze capacity and obtain tractable approximations of users' average data rates. The analytical results are validated by Monte-Carlos simulations.

2) Based on the analytical results, we further discuss the problem how to maximize the average data rate of UAV users (UUEs) while guaranteeing the performance of BUEs. The result has shown that power adjustment has more significant effect on the performance than height under our proposed model.

\vspace{-6pt}
\section{System Model}

We consider a downlink two-tier heterogeneous network as shown in Fig. 1, where $u$ and $b$ represent the UAV tier and BS tier respectively. The locations of UAVs and BSs are modeled as MHP ${\Phi _u}=\{x_0^u,x_1^u,x_2^u,...\}$ with density ${\lambda _u}$ at the same altitude $h$ and PPP ${\Phi _b}=\{x_0^b,x_1^b,x_2^b,...\}$ with density ${\lambda _b}$.  
 The UAVs have transmit power $P_u$ and each UAV serves its nearest single-antenna user. The UUEs are distributed according to another independent PPP.  Hence, the distance between a user and its serving UAV follows the distribution ${f_{r_u}}\left( r \right)= 2\pi {\lambda _u}{r}{e^{ - \pi {\lambda _u}\left( {{r^2} - {h^2}} \right)}}$. The BSs have transmit power $P_b$ and each BS serves a single-antenna user, who is uniformly distributed in the circle region centered at $x_i^b$ with radius $R_b$. Therefore, the distance between a user and its serving BS follows the distribution  ${f_{r_b}}\left( r \right) = \frac{{2r}}{{R_b^2}},0\le r \le R_b$ \cite{Sto2018Chen}. Considering the scarcity of frequency resource, universal frequency reuse is assumed, resulting in that BUEs may suffer significant interference from nearby UAVs. Therefore, UAVs should reduce the transmit power to ${\eta {P_u}}$, where $\eta  \in \left[ {0,1} \right]$ is the power control factor.

For the ground-to-ground (G2G) transmission, we adopt the standard path loss model with path loss exponent $\alpha_b > 2$. The small scale fading is assumed to be Rayleigh fading with unit mean. The A2G channel characteristics significantly different from G2G channels, 
depending on the environment, A2G links can be LoS or non-LoS (NLoS).
The expression for the LoS probability has been obtained in \cite{aal2014optimal}:
  \begin{small}
 \begin{equation}
{P_{l}(r)} = \frac{1}{{1 +  C\exp \left[ { - B\left( {\theta  - C} \right)} \right]}},
  \end{equation}
  \end{small}
where $\theta  = \frac{{180}}{\pi }\arcsin \left( {\frac{h}{r}} \right)$ denotes the elevation angle of a A2G link, $B$ and $C$ are environment constants. Furthermore, the probability of NLoS channel is ${P_{n}}(r) = 1 - {P_{l}}(r)$. Therefore, the set of UAVs ${\Phi _u}$ can be decomposed into two independent sets, i.e.,${\Phi _u}= {\Phi_l} \bigcup {\Phi_n}$, where ${\Phi_l}=\{x_0^l,x_1^l,x_2^l,...\}$ and ${\Phi_n}=\{x_0^n,x_1^n,x_2^n,...\}$ denote the set of LoS and NLoS UAVs, respectively. According to this, the path loss exponents for LoS and NLoS links are $ \alpha_l$ and $\alpha_n$. As for small scale fading, Nakagami-m fading is assumed.
The LoS link and the NLoS link have their own  parameters $m_l$ and $m_n$ \cite{w2008mathm}. Since the networks are typically interference limited, we ignore noise in this paper.

 Based on the system model, we first define the signal-to-interference (SIR) of users. Due to the stationary of MHP and PPP, we consider the performance of a typical BUE or UUE located at the origin in this paper.

 1) The SIR received by a typical BUE can be expressed as:
 \begin{small}
  \begin{align}
SI{R_0^b} &= \frac{{g_{00}^{bb}{\|x_0^b\|^{ - {\alpha _b}}}}}{{{I_u^b} + {I_b}}},
 \end{align}
 \end{small}
 \begin{figure}[!h]
\setlength{\abovecaptionskip}{-5pt}
\setlength{\belowcaptionskip}{10pt}
\centering
\includegraphics[width= 0.4\textwidth]{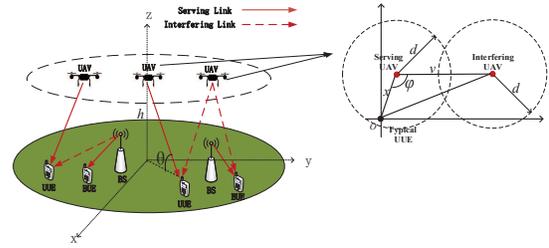}
\DeclareGraphicsExtensions.
\caption{Illustration of the system model. The red points are the projections of UAVs onto the ground, x is the horizontal distance between the serving UAV and typical UUE, v is the distance between serving and interfering UAV, d is the minimum distance among UAVs.}
\label{picture1}
\vspace{-10pt}
\end{figure}

where ${I_u^b}$ is the total interference from the UAVs and ${I_b}$ is the total interference from other BSs.

\begin{small}
 \begin{align}
 &{I_u^b}= \frac{{\eta {P_u}}}{{{P_b}}}\left( {\sum\limits_{x_i^l \in {\Phi _l}} {{\|x_i^l\|}^{ - {\alpha _l}}g_{i0}^{lb}}  + \sum\limits_{x_j^n \in {\Phi _n}} {{\|x_j^n\|}^{ - {\alpha _n}}g_{j0}^{nb}} } \right),\\
&{I_{b}}= \sum\limits_{{x_k^b} \in {\Phi _b}\backslash \left\{ x_0^b \right\}} {{\|x_k^b\|}^{ - {\alpha _b}}g_{k0}^{bb}},
 \end{align}
 \end{small}

 where $g_{00}^{bb}\sim \exp (1)$ and $g_{k0}^{bb}\sim \exp (1)$ are the channel gains between a typical BUE and its serving BS ${x_0^b} $ and other interfering BSs ${x_k^b} \in {\Phi_b}\backslash {\{x_0^b\}}$ \cite{Adaptive2010zhang}. And  $ g_{i0}^{lb}$, $ g_{j0}^{nb}$ are the gains between a typical BUE and LoS UAVs ${x_i^l} \in {\Phi_l}$ and NLoS UAVs ${x_j^n}\in {\Phi_n}$ which are follows Gamma distribution with probability density function (PDF) given by \cite{w2008mathm} with different parameters $m_l$ and $m_n$:
\begin{small}
\begin{equation}
 \label{nakagami}
{f_{g_{x_i^u}}}\left( g \right) = \frac{{{m^m}{g^{m - 1}}}}{{\Gamma \left( m \right)}}{e^{ - mg}}, \forall x_i^u \in {\Phi _u},
\end{equation}
\end{small}
where ${\Gamma \left( m \right)}$ is given by ${\Gamma \left( m \right)}=  \int\limits_0^{\infty} {{x^{m - 1}}{{ e }^{-x}}} dx$.

2) The SIR received by a typical UUE can be expressed as:
 \begin{small}
  \begin{align}
SI{R_0^s} &= \frac{{g_{00}^{ss}{\|x_0^s\|^{ - {\alpha _s}}}}}{{{I_b^s} + {I_u}}},
 \end{align}
 \end{small}
where ${I_b^s}$ is the total interference from the BSs and ${I_u}$ is the total interference from other UAVs:
\begin{small}
\begin{align}
&{I_{u}}= {\sum\limits_{x_i^s \in {\Phi _s} \backslash \left\{ x_0^s \right\}} {{\|x_i^s\|}^{ - {\alpha _s}}g_{i0}^{ss}}  + \sum\limits_{x_j^{s'} \in {\Phi _{s'}}} {{\|x_j^{s'}\|}^{ - {\alpha _{s'}}}g_{j0}^{{s'}s}} },\\
&{I_b^s}= \frac{{{P_b}}}{{\eta {P_u}}}\sum\limits_{{x_k^b} \in {\Phi _b}} {{\|x_k^b\|}^{ - {\alpha _b}}g_{k0}^{bs}},
\end{align}
\end{small}
 where $s \in \left\{ {l,n} \right\},s' = \left\{ {l,n} \right\}\backslash{\{s\}}$. $g_{k0}^{bs}\sim \exp (1)$ is the channel gain between a typical UUE and BSs ${x_k^b} \in {\Phi_b}$. And $g_{00}^{ss}$, $g_{i0}^{ss}$ and $g_{j0}^{{s'}s}$ are the gains between a typical UUE and its serving UAV ${x_0^s} $ and other UAVs ( ${x_i^s} \in{\Phi_s} \backslash {\{x_0^s\}}$ , ${x_j^{s'}} \in{\Phi_{s'}}$ ) which are follows Gamma distribution with PDF in (\ref{nakagami}) with parameters $m_s$.


The system area spectrum efficiency (ASE)£¬ which is defined as the average throughput per Hz per unit area :
\begin{small}
\begin{equation}
{\rm ASE}={\lambda_u}{R_u}+{\lambda_b}{R_B},
\end{equation}
\end{small}
where $ R_u = {\mathbb{E}}[\rm{log}(1 + \rm{SIR_0^s})]$ and  $ R_B = {\mathbb{E}}[\rm{log}(1 + \rm{SIR_0^b})]$ denote the typical UUE's and typical BUE's average data rate, respectively.
\vspace{-6pt}
\section{Analysis of the Average Data Rate}
Before deriving the approximate expression of users' average data rate, some preliminary results are provided.

\vspace{-6pt}
\subsection{ Mat\'{e}rn Hardcore Point Process}

 MHP is a kind of repulsive point process which is thinned from a independent PPP ${\Phi _p}=\{x_0^p,x_1^p,x_2^p,...\}$ with density ${\lambda _p}$ \cite{am2013coverage}. 
 The density of MHP is
\begin{small}
  \begin{equation}
 {\lambda _u}{\rm{ = }}p{\lambda _p} = \frac{{1{\rm{ - }}\exp ( - {\lambda _p}\pi {d^2})}}{{\pi {d^2}}}.
\end{equation}
\end{small}
 The second order product density of MHP ${\Phi _u}$ is given by
  \begin{small}
\begin{equation}
  \label{pdf}
  {\chi ^{\left( 2 \right)}}\left( v \right){\rm{ = }}\left\{ \begin{array}{l}
\frac{{2V\left( v \right)\left[ {1 - e^{ - {\lambda _p}\pi {d^2}} } \right]}-{2\pi {d^2}\left[ {1 -  e^{ - {\lambda _p}V(v)}} \right]} }{{\pi {d^2}V\left( v \right)\left[ {V\left( v \right) - \pi {d^2}} \right]}} ,d \le v \le 2d{\kern 1pt} {\kern 1pt} \\
 {\lambda _u}^2,{\kern 1pt} v \ge 2d\\
0,otherwise{\kern 1pt},
\end{array} \right.
  \end{equation}
\end{small}

where $V(v)$ is the area of the union of two circles as Fig. 1, it can be calculated as follow:
 \begin{small}
 \begin{equation}
V(v) = 2\pi {d^2} - {\kern 1pt} 2{d^2}{\cos ^{ - 1}}(\frac{v}{{2d}}) + v\sqrt {{d^2} - \frac{{{v^2}}}{4}}.
  \end{equation}
\end{small}



\vspace{-6pt}
\subsection{ MISR-based Gain in MHP}
In terms of MHP, performance analysis is not a easy task. In order to address the problem, we adopt MISR-based gain method \cite{m2014themean} and the gain is calculated as:
 \begin{small}
\begin{equation}
 \label{gain}
\rm{G} = \frac{{\rm{MIS}{\rm{R}_{\rm{PPP}}}}}{{\rm{MIS}{\rm{R}}}},
 \end{equation}
 \end{small}
where MISR is defined as $\rm{MISR} \buildrel \Delta \over = \mathbb{E}\left( {\frac{1}{{{\mathbb{E}_{\rm{g}}}\left(\rm{S} \right)}}} \right)$ and ${\mathbb{E}_{\rm{g}}}\left( \rm{S} \right) $ is the signal power averaged over the fading. To obtain the MISR-based gain, we need first derive the expressions of ${\rm{MIS}}{{\rm{R}}_{\rm{PPP}}}$ and $ {\rm{MIS}}{{\rm{R}}_{\rm{MHP}}}$ in Theorem 1.

\begin{thm}
 The expressions of ${\rm{MIS}}{{\rm{R}}_{\rm{PPP}}}$ and $ {\rm{MIS}}{{\rm{R}}_{\rm{MHP}}}$  are:
 \end{thm}
  \begin{small}
\begin{align}
{\rm{MIS}}{{\rm{R}}_{{\rm{PPP}}}} &= \sum\limits_{s \in \left\{ {l,n} \right\}} {\int\limits_{\rm{h}}^\infty  {{P_s}({\rm{r}}){{\rm{{\mathbb{E}}}}_{p,s}}\left( {{\rm{I\bar SR}}\left| r \right.} \right){f_{r_p}}\left( r \right)dr} },\\
{\rm{MIS}}{{\rm{R}}_{{\rm{MHP}}}} &= \sum\limits_{s \in \left\{ {l,n} \right\}} {\int\limits_{\rm{h}}^\infty  {{P_s}({\rm{r}}){{\rm{{\mathbb{E}}}}_{u,s}}\left( {{\rm{I\bar SR}}\left| r \right.} \right){f_{r_u}}\left( r \right)dr} },
 \end{align}
\end{small}
with
 \begin{small}
\begin{align}
\mathop {{{\rm{{\mathbb{E}}}}_{\rm{p},\rm{s}}}\left( {{\rm{I\bar SR}}\left| r \right.} \right)} &= 2\pi {\lambda _p}{r^{{\alpha _{\rm{s}}}}}\sum\limits_{q \in \{ l,n\} } {\int_r^\infty  {{P_q}(y){y^{1 - {\alpha _q}}}} dy}, \\
\mathop {{{\rm{{\mathbb{E}}}}_{\rm{u},\rm{s}}}\left( {{\rm{I\bar SR}}\left| r \right.} \right)} &= \frac{{{r^{{\alpha _s}}}}}{{{\lambda _u}}}\int\limits_0^{2\pi } {\int\limits_{{v_1}}^{{v_2}} {F(x,v,\varphi )} } \chi _1^{(2)}\left( v \right)dvd\varphi \nonumber \\
&+ \frac{{{r^{{\alpha _s}}}}}{{{\lambda _u}}}\int\limits_0^{2\pi } {\int\limits_{{v_2}}^\infty  {F(x,v,\varphi )} \chi _2^{(2)}\left( v \right)dvd\varphi },
 \end{align}
 \end{small}
and $F(x,v,\varphi )$ is given below:
 \begin{small}
 \begin{equation}
F(x,v,\varphi ) = \sum\limits_{q \in \left\{ {l,n} \right\}} {\frac{{{P_q}(v)v}}{{{{(\sqrt {{v^2} + {{\rm{r}}^2} - 2xv\cos \left( \varphi  \right)} )}^{{\alpha _q}}}}}},
\end{equation}
\end{small}
where  $x \!\!\!=\!\!\! \sqrt {{r^2} - {h^2}}$. The integral boundary ${v_1} = \max \left[ {d,2x\left| {\cos (\varphi )} \right|} \right]$, ${v_2} = \max \left[ {2d,2x\left| {\cos (\varphi )} \right|} \right]$. For simplicity of notation, the expression of ${\chi ^{\left( 2 \right)}}\left( v \right)$ in (\ref{pdf}) simplify as:

 \begin{small}
 \begin{align}
{\chi ^{\left( 2 \right)}}\left( v \right) = \left\{ {\begin{array}{*{20}{c}}
{\chi _1^{(2)}\left( v \right),}&{if{\kern 1pt} {\kern 1pt} d \le v \le 2d}\\
{\chi _2^{(2)}\left( v \right),}&{if{\kern 1pt} {\kern 1pt} {\kern 1pt} v \ge 2d}.
\end{array}} \right.
\end{align}
\end{small}
\begin{proof}
See Appendix A.
\end{proof}

\vspace{-6pt}
\subsection{Stochastic Geometry Analysis}

\begin{figure*}
\centering
\begin{minipage}[t]{0.33\linewidth}
\centering
\includegraphics[width=1.9in]{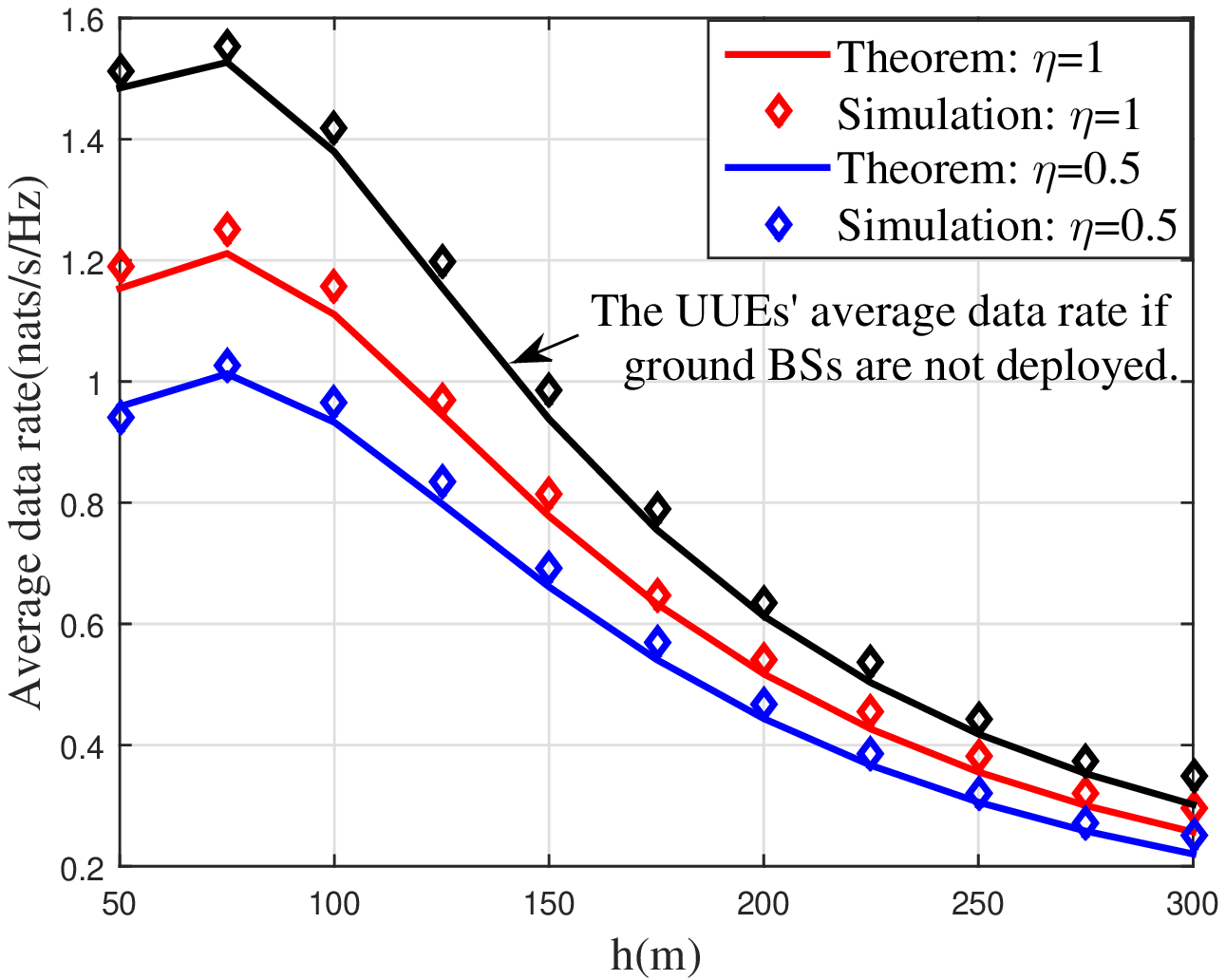}
\caption{Average data rate of UUEs. }
\label{fig:side:a}
\end{minipage}%
\begin{minipage}[t]{0.33\linewidth}
\centering
\includegraphics[width=1.9in]{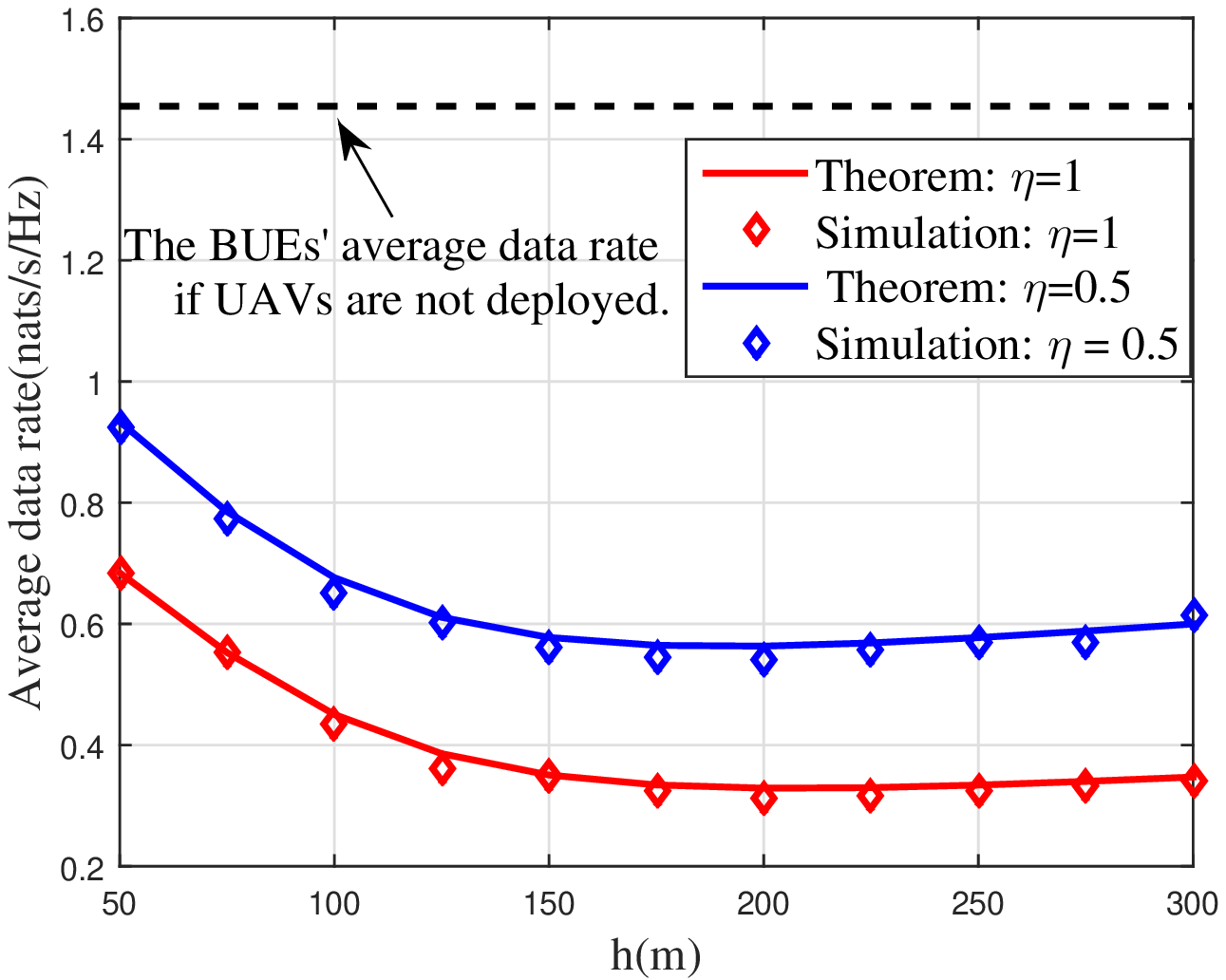}
\caption{Average data rate of BUEs.}
\label{fig:side:b}
\end{minipage}
\begin{minipage}[t]{0.33\linewidth}
\centering
\includegraphics[width=1.9in]{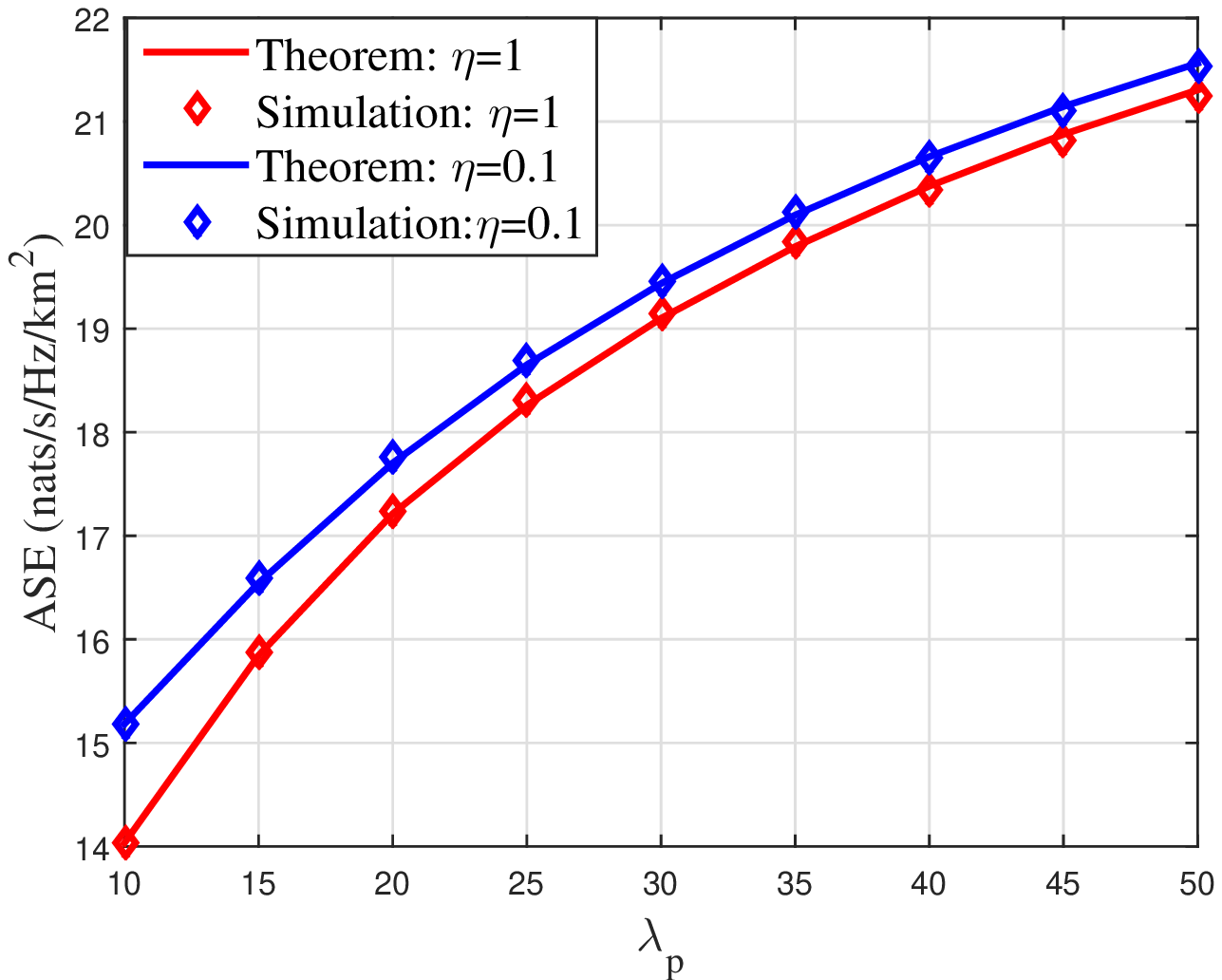}
\caption{System ASE.}
\label{fig:side:c}
\end{minipage}
\end{figure*}
With the help of capacity lemma \cite{ka2010lemma}, we can obtain the approximations of ${{\hat R}_u}$ and ${{\hat R}_B}$ by leveraging the gain (\ref{gain}) in Theorem 2.
\begin{thm}
 The approximations of ${{\hat R}_u}$ and ${{\hat R}_B}$ are:
 \end{thm}
  \begin{small}
 \begin{align}
&{{\hat R}_u} \!=\!\!\! \sum\limits_{s \in \left\{ {l,n} \right\}} {\int\limits_0^\infty  {\int\limits_h^\infty  {\frac{{1 \!\!-\!\! {{\left( {1 \!+\! \frac{z}{m_s}} \right)}^{ \!-\! {m_s}}}}}{z}} } }
{P_s}\left( r \right){e^{ - {K_s}(r,z)}}{f_{r_p}(r)}drdz, \\
&{{\hat R}_B}= \int\limits_0^\infty  {\int\limits_0^{{R_b}} {\frac{{1 - {{\left( {z + 1} \right)}^{ - 1}}}}{z}{{\rm{e}}^{ - {{\rm{K}}_b}\left( {r,z} \right)}}} } \frac{{2r}}{{R_b^2}}drdz,
 \end{align}
 \end{small}

where ${{K_s}(r,z)}$ and ${{{\rm K}_b}\left( {r,z} \right)}$ are given below:
 \begin{small}
\begin{align}
{K_s}(r,z)&= 2\pi {\lambda _p}\int\limits_r^\infty  {{\kappa _s}\left( {r,y,z} \right)ydy} \nonumber \\
 &+ \frac{{2\pi {\lambda _b}}}{{{\alpha _b}}}{\left( {\frac{{z{P_b}{r^{{\alpha _s}}}}}{{\eta {P_u}}}} \right)^{\frac{2}{{{\alpha _b}}}}}B\left( { \frac{2}{{{\alpha _b}}},1 -\frac{2}{{{\alpha _b}}}} \right),
\end{align}

\begin{equation}
\begin{array}{l}
{K_b}\left( {r,z} \right)\!= \!2\pi {\lambda _u}\int\limits_h^\infty  {{\kappa _b}\left( {r,y,z} \right)ydy}  \!+\! \frac{{2\pi {\lambda _b}}{ z^{\frac{2}{{{\alpha _b}}}}r^2}}{{{\alpha _b}}}B\left( { \frac{2}{{{\alpha _b}}},1 \!\!-\!\! \frac{2}{{{\alpha _b}}}} \right),
\end{array}
 \end{equation}
  \end{small}
with
 \begin{small}
 \begin{equation}
{\kappa _s}\left( {r,y,z} \right) \!=\! \sum\limits_{q \in \left\{ {l,n} \right\}} {\left[ {1 \!-\! {{\left( {1 \!+\! z\frac{{{y^{ \!-\! {\alpha _{_q}}}}{r^{{\alpha _s}}}}}{{{m_q}G}}} \right)}^{ \!-\! {m_q}}}} \right]{P_q}\left( y \right)},
 \end{equation}

\begin{equation}
{\kappa _b}\left( {r,y,z} \right) \!=\! \sum\limits_{q \in \left\{ {l,n} \right\}}
{\left[ {1 \!-\! {{\left( {1 \!+\! z\frac{{\eta {P_u}{r^{{\alpha _B}}}}}{{{m_q}P_b{y^{  {\alpha _{_q}}}}}}} \right)}^{ - {m_q}}}} \right]{P_q}\left( y \right)},
\end{equation}
\end{small}
where  $B\left( {P,Q} \right) = \int\limits_0^1 {{x^{P - 1}}{{\left( {1 - x} \right)}^{Q - 1}}} dx $.

\begin{proof}
See Appendix B.
\end{proof}

The expression of the system ASE is omitted for space saving.

\vspace{-6pt}
\begin{table}[!h]
\centering
\renewcommand{\tablename}
\caption{\centering{ TABLE I} \protect \\  NUMERICAL/SIMULATION PARAMETERS}\\


\begin{tabular}{|c|c|c|c|c|c|}
\hline
Symbol &  Value & Symbol &  Value& Symbol &  Value \\
\hline
${\lambda _u}$ & $10k{m^{ - 2}}$ & ${P_u}$ & $37dBm$ & $d$ &$100m$ \\
\hline
${\lambda _b}$ & $10k{m^{ - 2}}$ & ${P_b}$ & $37dBm$ &${\alpha _b}$ & $4$ \\
\hline
${\alpha _l}$ & $3$ & ${\alpha _n}$& $4$ &$B$ & $0.136$\\
\hline
 $m_l$ & $3$ & $m_n$ & $1$  & $C$ & $11.95$\\
\hline
\end{tabular}
\label{tab1}
\vspace{-8pt}
\end{table}


We provide simulations to evaluate our analytical results in Theorem 2 by MATLAB. Considering the propagation environment in dense urban area \cite{aal2014optimal}, system parameters are set according to Table I. Fig. 2 and Fig. 3 illustrate the UUEs' and BUEs' average data rate, respectively. We can observe that the simulation results match the analytical results quite well. Additionally, we depict the UUEs' and BUEs' average data rate in a single-tier network in Fig. 2 and Fig. 3, respectively.


Fig. 4 describes the system ASE as a function of $\lambda_u$, we can find that deploying UAVs can improve the system ASE. However, the deployment of UAVs will bring severe inter-tier interference to BUEs additionally, leading to the performance degradation of BUEs. Based on the premise of protecting BUEs' performance, we can improve the BUEs' performance at the expense of UUEs'. Thus, it is necessary to discuss the optimal parameters of UAVs to achieve a tradeoff between the UUEs' and BUEs' performance.

\vspace{-6pt}
\section{Discussions on Optimal Parameters of UAVs}
 In this section, we discuss the optimal parameters to achieve a tradeoff between the UUEs' and BUEs' average data rate.
\vspace{-6pt}
\subsection{Problem Formulation}

In order to utilize UAVs efficiently while guaranteeing the performance of BUEs, it is necessary to investigate the optimal parameters of UAVs. In other words, we should find the optimal $\left( {\eta ,h} \right)$ to maximize ${{\hat R}_u}(\eta ,h)$ under the constrain ${{\hat R}_B}(\eta ,h) \ge {R_{th}}$, where ${R_{th}} $ is the target of BUEs' average data rate. Towards this end, we formulate an optimization problem as follows:
  \begin{small}
 \begin{align}
P_0:&\mathop{\max }\limits_{h;\eta \in [0,1]} {\kern 1pt} {\kern 1pt} {\kern 1pt} {{\hat R}_u}(\eta ,h)\nonumber\\
 &subject{\kern 1pt} {\kern 1pt} to{\kern 1pt} {{\hat R}_B}(\eta ,h) \ge {R_{th}}.
\end{align}
 \end{small}
 \begin{thm}
Based on the Leibnitz integral rule, we can prove that ${{\hat R}_u}(\eta ,h)$ increases with $\eta$ but ${{\hat R}_B}(\eta ,h)$ decreases with $\eta$ monotonically.
 \end{thm}

\begin{proof}
See Appendix C.
\end{proof}
Thus, it is not difficult to show that the inequality constraint in $P_0$ can be replaced by the equality constrain ${\hat R_B}(\eta ,h){\rm{ = }}{R_{th}}$ using reduction to absurdity.

\vspace{-6pt}
\subsection{The Optimal Parameters of UAVs}
In this subsection, we focus on the optimal solution to ${P_0}$.  Firstly, we exhaust all the values of $h$ in the range ${\left[ h_{min},h_{max} \right ] }$ with a certain precision, where ${h_{min}}$, ${h_{max}}$ are set as $50$ and $300$ to determine the search range of $h$.  Then, based on the Theorem 3, we can solve ${\hat R_B}(\eta ,h){\rm{ = }}{R_{th}}$ by binary search and obtain all the combinations of $(\eta ,h)$.  Finally, we select $(\eta^* ,h^*)$ as optimal result that maxmize ${{\hat R}_u}(\eta ,h)$. The optimal result is also shown in table II.
 \begin{table}[!h]
\centering
\renewcommand{\tablename}
\caption{\centering{ TABLE II} \protect \\  THE OPTIMAL SOLUTION TO ${P_0}$}\\
\label{tab2}

\begin{tabular}{|c|c|c|c|c|c|c|c|}
\hline
${R_{th}}\left( {nats/s/Hz} \right)$ &0.6&0.7&0.8&0.9&1.0&1.1\\
\hline
${h^ * }\left( m \right)$ &70&50&50&50&50&50\\
\hline
${\eta ^ * }$&0.87&0.92&0.53&0.30&0.17&0.09\\
\hline
${\hat R_u^ *}\left( {nats/s/Hz} \right)$ &1.20&1.14&1.07&0.99&0.90&0.78\\
\hline

\end{tabular}
\vspace{-6pt}
\end{table}

 As the UAVs should sacrifice their own performance to guarantee the performance of the BUEs, we observe that ${\hat R_u}$ monotonously decreases with ${R_{th}}$. What's more, we notice that the optimal $\eta^*$ decreases with higher demand of ${R_{th}}$ and the optimal $ h^*$ is always the the lowest among all the combinations satisfying BUEs' target rate except the case of ${R_{th}=0.7}$. The reason is  that we can obtain higher ${\hat R_u^*}$ with lower power control factor due to the existence of the optimal altitude for UUEs' average data rate. The reason of the general trend of $\eta^*$ with ${R_{th}}$ is that the inter-tier interference is less serious with lower transmit power of UAVs which is used to satisfy higher target BUEs' average data rate. According to the results, we can easily observe that it is more effective to adjust UAVs' transmit power than the height to maximize the UUEs' average data rate. Thus, the maximal UUEs' average data rate is achieved mainly by power control factor of UAVs under our proposed model.
\vspace{-6pt}
\section{Conclusion}
In this paper, we studied statistic performance considering a two-tier UAV-assisted network. Based on MISR-based gain method, we have derived the concise approximations of the users' average data rate. The accuracy of the theoretical results have been verified through simulations. With the help of the more tractable approximations, the properties of users' average data rate are observed. Moreover, we have discussed the optimal combinations of height and power control factor of UAVs to maximize UUEs' average data rate while satisfying BUEs' target average data rate. The result has shown that power control factor has more significant effect on the performance than height under our proposed model.

\vspace{-6pt}
\appendix

\vspace{-6pt}
\subsection{Proof of Theorem 1 \label{appendix:A}}
\setcounter{equation}{0}
\renewcommand{\theequation}{A.\arabic{equation}}
The expression for ${{\rm{{\mathbb{E}}}}_{p,s}}\left( {{\rm{I\bar SR}}\left| r \right.} \right)$ and ${{\rm{{\mathbb{E}}}}_{u,s}}\left( {{\rm{I\bar SR}}\left| r \right.} \right)$ where $s \!\in\! \{l,n\}$ in PPP and MHP model for UAV tier are obtained as follws. As for PPP ${\Phi _p}=\{x_0^p,x_1^p,x_2^p,...\}$, it also divided into two independent sets, i.e., ${\Phi _p}= {\Phi_{p,l}} \bigcup {\Phi_{p,n}}$, where ${\Phi_{p,l}}=\{x_0^{p,l},x_1^{p,l},x_2^{p,l},...\}$ and  ${\Phi_{p,n}}=\{x_0^{p,n},x_1^{p,n},x_2^{p,n},...\}$.  For simplicity of notation, we denote the distance of serving link as r,  the distance interfering link as y. Assuming the serving A2G channel is LoS channel:
 \begin{small}
 \begin{align}
{{{\rm{{\mathbb{E}}}}_{p,l}}\left( {{\rm{I\bar SR}}\left| r \right.} \right)}
&{ \!=\!  \left[ {\sum\limits_{x_i^{p,l}\!  \in\!   {\Phi _{p,l}}\backslash \left\{ x_0^{p,l} \right\}} \frac{\|x_0^{p,l}\|^{\alpha _l}}{\|x_i^{p,l}\|^{  {\alpha _l}}} \!  +\!  \sum\limits_{x_j^{p,n} \! \in\!  {\Phi _{p,n}}} \frac{\|x_0^{p,l}\|^{\alpha _l}}{{\|x_j^{p,n}\|}^{ {\alpha _n}}} } \right]}\nonumber \\
&{\mathop  = \limits^{\left( a \right)} \sum\limits_{q \in \left\{ {l,n} \right\}} {2\pi {\lambda _p}{r^{{\alpha _l}}}\int_r^\infty  {{P_q}(y){y^{1 - {\alpha _q}}}} dy}},
\end{align}
\end{small}
 where $(a)$ follows Campbell's theorem. 


 For a MHP $\Phi_u$, we first analyze the case of the serving link is LoS link, the expression of ${{\rm{{\mathbb{E}}}}_{u,l}}$ is:
 \begin{small}
 \begin{align}
&{{{\rm{{\mathbb{E}}}}_{u,l}}\left( {{\rm{I\bar SR}}\left| r \right.} \right)}{= {{\sum _{ x_i^{l} \in {\Phi _{l}}\backslash \left\{  x_0^{l}\right\}}}{{\left( {\frac{\|x_0^{l}\|}{{ \|x_i^{l}\|}}} \right)}^{{\alpha _l}}} \!\!\!+\!\!\! \sum\limits_{x_j^{n} \in {\Phi _{n}}} { {\frac{{{\|x_0^{l}\|^{{\alpha _l}}}}}{{{ \|x_j^{n}\|}^{{\alpha _n}}}}} } } }\nonumber \\
&\mathop = \limits^{\left( b \right)} \frac{{{r^{{\alpha _l}}}}}{{{\lambda _u}}}\left[ {\int\limits_0^{2\pi } {\int\limits_{{v_1}}^{{v_2}} {\sum\limits_{q \in \left\{ {l,n} \right\}} {\frac{{{P_q}(v){\chi _1^{(2)}\left( v \right)}v}}{{{{(\sqrt {{v^2} + {{\rm{r}}^2} - 2xv\cos \left( \varphi  \right)} )}^{{\alpha _q}}}}}}} } d\varphi dv} \right.\nonumber \\
&+ \left. {\int\limits_0^{2\pi } {\int\limits_{{v_2}}^\infty  {\sum\limits_{q \in \left\{ {l,n} \right\}} {\frac{{{P_q}(v){\chi _2^{(2)}\left( v \right)}v}}{{{{(\sqrt {{v^2} + {{\rm{r}}^2} - 2xv\cos \left( \varphi  \right)} )}^{{\alpha _q}}}}}}} } dvd\varphi } \right],
\end{align}
\end{small}
where $(b)$ also follows the Campbell's theorem with PDF in (\ref{pdf}). The derivation of  ${{\rm{{\mathbb{E}}}}_{p,n}}\left( {{\rm{I\bar SR}}\left| r \right.} \right)$, ${{{\rm{{\mathbb{E}}}}_{u,n}}\left( {{\rm{I\bar SR}}\left| r \right.} \right)}$ is the same as ${{\rm{{\mathbb{E}}}}_{p,l}}\left( {{\rm{I\bar SR}}\left| r \right.} \right)$, ${{{\rm{{\mathbb{E}}}}_{u,l}}\left( {{\rm{I\bar SR}}\left| r \right.} \right)}$expect that ${r^{{\alpha _l}}}$ is replaced by ${r^{{\alpha _n}}}$.
\vspace{-10pt}
\subsection{Proof of Theorem 2 \label{appendix:B}}
\setcounter{equation}{0}
\renewcommand{\theequation}{B.\arabic{equation}}
In order to derive the expression for the users' average  data rate, we introduce the capacity calculation lemma in \cite{ka2010lemma}:
 \begin{small}
 \begin{equation}
 \label{rate}
 \mathbb{E}\left[\rm{log} \left(1 + \frac{X}{Y}\right)\right]= \int\limits_0^\infty  {\frac{1}{z}\mathbb{E}\left[ e^{-zY}\right] \left( {1 - \mathbb{E}\left[ e^{-zX}\right]} \right)} dz.
 \end{equation}
 \end{small}
Based on the MISR-based gain, the SIR distribution of non-PPP network can be accurately approximated by that of a PPP through scaling the threshold z,
 i.e.
  \begin{small}
\begin{align}
{\mathbb{P}}\left( SIR_0^s >z \right)&= {\mathbb{P}}\left( g_{00}^{ss} >z{\|x_0^{p, s}\|^{  {\alpha _s}}}\left({I_b^s} + {I_u}\right) \right)\nonumber\\
 &\approx \mathbb{P}\left( g_{00}^{ss} >z{\|x_0^{p, s}\|^{  {\alpha _s}}}\left({I_b^s} + {I_u^{\rm{PPP}}}/{G}\right) \right),
\end{align}
\end{small}
with

  \begin{small}
\begin{align}
{I_u^{\rm{PPP}}}= {\sum\limits_{x_i^s \in {\Phi _{p,s}} \backslash \left\{ x_0^s \right\}} {{\|x_i^s\|}^{ - {\alpha _s}}g_{i0}^{ss}}  + \sum\limits_{x_j^{p,s'} \in {\Phi _{p,s'}}} {{\|x_j^{p,s'}\|}^{ - {\alpha _{s'}}}g_{j0}^{{s'}s}} },
\end{align}
\end{small}
   According to the relationship between the SIR distribution and capacity, we have ${{{\hat R}_u}} = \sum\limits_{s \in \left\{ {l,n} \right\}}\int\limits_0^\infty  {\frac{{\mathbb{P}}\left( SIR_0^s >z \right)}{{1 + z}}dz} $, thus:
 \begin{small}
 \begin{align}
{{{\hat R}_u}}  = \sum\limits_{s \in \left\{ {l,n} \right\}} \int\limits_0^\infty  {\frac{{{1 \!-\! {\mathbb{E}}\left[ {{e^{ \!-\! zg_{00}^{ss}}}} \right]}}}{z}  {{\mathcal{L}}_{I,s}} \left( z \right)  dz}.
\end{align}
\end{small}
with ${{\mathcal{L}}_{I,s}} \left( z \right)=\mathbb{E} \left[e^{-z{\|x_0^{p, s}\|^{  {\alpha _s}}}\left({I_b^s+ {I_u^{\rm{PPP}}}/{G}}\right) }\right]$.
 Then, we deduce the expression of ${{\mathbb{E}}\left[ {{e^{ - zg_{00}^{ss}}}} \right]}$ and ${{\mathcal{L}}_{I,s}} \left( z \right)$.
${\mathbb{E}}\left[ {e^{ - zg_{00}^{ss}}} \right]= {\left( {z/{m_s} + 1} \right)^{ - {m_s}}}$,
where $g_{00}^{ss}$ follows the PDF in (\ref{nakagami}) with parameter $m_s$. Assuming that serving UAV is LoS channel, the approximation of ${{\mathcal{L}}_{I,l}} \left( z \right)$ can be calculated below: 
 \begin{small}
\begin{align}
&{{\mathcal{L}}_{I,l}} \left( z \right){ = \mathop {{\rm{ }}{\mathbb{E}}}\limits_{x_i^{p,l} \!\in\! {\Phi _{p,l}}\backslash \left\{ x_0^{p,l} \right\}} \left[ {{e^{ - \frac{{\|x_0^{p,l}\|^{{\alpha _l}}}{g_{i0}^{ll}}}{{\|x_i^{p,l}\|}^{ \alpha _l}{G}}}}} \right] }\nonumber \\
&{\times \mathop {{\rm{ }}{\mathbb{E}}}\limits_{x_j^{p,n}\in {\Phi _{p,n}}} \left[ {{e^{ - \frac{{\|x_0^{p,l}\|^{{\alpha _l}}}{g_{j0}^{nl}}}{{\|x_j^{p,n}\|}^{ \alpha _n}G}}}} \right] }
{\times \mathop {{\rm{ }}{\mathbb{E}}}\limits_{x_k^{b}\in {\Phi _b}} \left[ {{e^{ - \frac{{{P_b}}{\|x_0^{p,l}\|^{{\alpha_l}}}g_{k0}^{bl}}{{\eta {P_u}}{\|x_k^{b}\|}^{ {\alpha_b}}}}}} \right]}\nonumber \\
&{\mathop {{\rm{  }} = }\limits^{\left(c \right)} \exp \left\{ { \!\!-\!2\pi {\lambda _p}\int\limits_r^\infty  {\!\!\sum\limits_{q \!\in\! \left\{ {l,n} \right\}} {\left[ {1 \!\!-\!\! {{\left( {1 \!\!+\!\! z\frac{{{y^{ \!-\! {\alpha _{_q}}}}{r^{{\alpha _l}}}}}{{{m_q}G}}} \right)}^{ \!\!-\!\! {m_q}}}} \right]{P_q}\left( y \right)}ydy} } \right\}}\nonumber \\
&{ \times \exp \left\{ { - \frac{{2\pi {\lambda _b}}}{{{\alpha _b}}}{{\left( {\frac{{z{P_b}{r^{{\alpha _l}}}}}{{\eta {P_u}}}} \right)}^{\frac{2}{{{\alpha _b}}}}}B\left( {\frac{2}{{{\alpha _b}}},1- \frac{2}{{{\alpha _b}}}} \right)} \right\}},
\label{interf}
\end{align}
 \end{small}
where (c) follows the probability generating functional of the aggregate interference. $g_{i0}^{ll}$, $g_{j0}^{nl}$ follows PDF in (\ref{nakagami}) with parameters  ${{m_l}}$, ${{m_n}}$ and $g_{k0}^{bl}\sim \exp(1)$. Above all, ${{\mathcal{L}}_{I,n}}(z) $ is the same as  (\ref{interf}) except that ${r^{{\alpha _l}}}$ is replaced by ${r^{{\alpha _n}}}$.

The BUEs' approximation average data rate is
 \begin{small}
 \begin{align}
{{\hat R}_B}= \int\limits_0^\infty  {\int\limits_0^{{R_b}} {\frac{{1 - {\mathbb{E}}\left[ {{e^{ - zg_{00}^{bb}}}} \right]}}{z}} }{{\mathcal{L}}_{I,b}}(z) {f_{r_b}}\left( r \right)drdz,
\end{align}
\end{small}
with ${{\mathcal{L}}_{I,b}}(z)=\mathbb{E} \left[e^{-z{\|x_0^b\|^{ - {\alpha _b}}}\left({{I_u^b} + {I_b}}\right) }\right]$. Then we deduce the expression of ${{\mathbb{E}}\left[ {{e^{ - zg_{00}^{bb}}}} \right]}$ and ${{\mathcal{L}}_{I,b}}(z)$,
${\mathbb{E}}\left[ {{e^{ - zg_{00}^{bb}}}} \right] = {\left( {z + 1} \right)^{ - 1}}$, where $g_{00}^{bb}\sim \exp(1)$.
The expression of ${{\mathcal{L}}_{I,b}}(z)$ is:
\begin{small}
\begin{align}
&{{\mathcal{L}}_{I,b}}(z) = \mathop{ \mathbb{E}}\limits_{x_i^l \in {\Phi _{l}}} \left[ {{e^{ - \frac{{\eta P_u}\|x_0^b\|^{{\alpha _b}}}{P_b\|x_i^l\|^{ {\alpha _l}}}g_{i0}^{lb}}}} \right] \nonumber \\
&\times \mathop { \mathbb{E}}\limits_{x_j^n \in {\Phi _{n}}} \left[ {{e^{ - \frac{{\eta P_u}\|x_0^b\|^{{\alpha _b}}}{P_b\|x_j^n\|^{ {\alpha _n}}}g_{j0}^{nb}}}} \right]
\times \mathop { \mathbb{E}}\limits_{x_k^b\in {\Phi _b}\backslash \left\{ x_0^b \right\}} \left[ {{e^{ - \frac{\|x_0^b\|^{{\alpha _b}}}{\|x_k^b\|^{ {\alpha _b}}}g_{k0}^{bb}}}} \right]\nonumber
 \end{align}
 \begin{align}
 &= \exp \left\{ {  \!- \! 2\pi {\lambda _u}\int\limits_h^\infty  {  \!\!\sum\limits_{q \in \left\{ {l,n} \right\}}
{  \!\!\left[ {1 \! \!-\! \! {{\left( {1  \!\!+\! \! z\frac{{\eta {P_u}{r^{{\alpha _b}}}}}{{{m_q}P_b{y^{  {\alpha _q}}}}}} \right)}^{  \!- \! {m_q}}}} \right]{P_q}\left( y \right)}ydy} } \right\}\nonumber \\
 &\times \exp \left[ { - \frac{{2\pi {\lambda _b}}}{{{\alpha _b}}}{z^{\frac{2}{{{\alpha _b}}}}}{r^2}B\left( {\frac{2}{{{\alpha _b}}},1 - \frac{2}{{{\alpha _b}}}} \right)} \right].
\end{align}
\end{small}
\subsection{Proof of Theorem 3 \label{appendix:C}}
\setcounter{equation}{0}
\renewcommand{\theequation}{C.\arabic{equation}}

In order to derive the monotonicity of ${{\hat R}_B}$ and ${{\hat R}_u}$ with $\eta$, we will differentiate ${{\hat R}_B}$ and ${{\hat R}_u}$ with respect to $\eta$, respectively. Based on Leibnitz integral rule, the analytical results are as follows:

\begin{small}
\begin{align}
\label{C1}
\frac{{\partial {{\hat R}_u}}}{{\partial \eta }} &= \sum\limits_{s \in \left\{ {l,n} \right\}} {\int\limits_0^\infty  {\int\limits_h^\infty  {\frac{{1 - {{\left( {1 + \frac{z}{{{m_s}}}} \right)}^{ - {m_s}}}}}{z}} } } {P_s}\left( r \right){e^{ - {K_s}(r,z)}}{f_{{r_p}}}(r) \nonumber \\
 &\times \frac{{2\pi {\lambda _b}}}{{{\alpha _b}}}{\left( {\frac{{z{P_b}{r^{{\alpha _s}}}}}{{{P_u}}}} \right)^{\frac{2}{{{\alpha _b}}}}}\frac{2}{{{\alpha _b}}}{\eta ^{ - \frac{2}{{{\alpha _b}}} - 1}}B\left( {\frac{2}{{{\alpha _b}}},1 - \frac{2}{{{\alpha _b}}}} \right)drdz
\end{align}
\end{small}

\begin{small}
\begin{align}
\label{C2}
&\frac{{\partial {{\hat R}_B}}}{{\partial \eta }} = {\rm{ - }}\int\limits_0^\infty  {\int\limits_0^{{R_b}} {\frac{{1 - {{\left( {z + 1} \right)}^{ - 1}}}}{z}{{\rm{e}}^{ - {{\rm{K}}_b}\left( {r,z} \right)}}} } \frac{{2r}}{{R_b^2}}\nonumber \\
 &\times \sum\limits_{q \in \left\{ {l,n} \right\}} {\int\limits_h^\infty  {2\pi {\lambda _u}\frac{{z{P_u}{r^{{\alpha _b}}}}}{{{P_b}}}{{\left( {1 + \frac{{z{P_u}{r^{{\alpha _b}}}\eta }}{{{m_q}{P_b}{y^{{\alpha _q}}}}}} \right)}^{ - {m_q} - 1}}{y^{1 - {\alpha _q}}}dy} } drdz
\end{align}
\end{small}
Based on the expressions above, we can easily find that the differential result (\ref{C1}) is positive and (\ref{C2}) is negative. Thus, ${{\hat R}_u}$ increases with ${\eta}$ and ${{\hat R}_B}$ decreases with $\eta$ monotonically.

\vspace{-6pt}
\bibliographystyle{IEEEtran}
\bibliography{IEEEabrv,myref}

\end{document}